\documentclass[11pt,letterpaper]{article}
\usepackage{graphicx}
\usepackage{epsfig,subfigure}
\usepackage{amssymb}
\usepackage{amsthm}
\usepackage{amsfonts}
\usepackage{amsmath}
\usepackage{multicol,multirow}
\usepackage[margin=1in]{geometry}

\newcommand{\cA}{\mathcal{A}}

\newcommand{\cN}{\mathcal{N}}

\newcommand{\cP}{\mathcal{P}}

\newcommand{\cR}{\mathcal{R}}

\newcommand{\cU}{\mathcal{U}}

\newcommand{\bP}{\mathbf{P}}

\newcommand{\tr}{\mathbf{Tr}}

\newcommand{\E}{\mathbf{E}}

\makeatletter
\def\namedlabel#1#2{\begingroup
   \fPtect\def\@currentlabel{#2}%
   \label{#1}\endgroup
}
\makeatother

\newcommand{\BlackBox}{\rule{1.5ex}{1.5ex}}
\renewenvironment{proof}{\par\noindent{\bfseries\upshape
  Proof\ }}{\hfill\BlackBox\\[2mm]}

\newtheorem{theorem}{Theorem}

\newtheorem{proposition}[theorem]{Proposition}

\begin{document}

\title{Resource Allocation for Statistical Estimation}
\author{Quentin Berthet$^\dag$ and Venkat Chandrasekaran$^{\dag,\ddag}$ \thanks{Email: qberthet@caltech.edu, venkatc@caltech.edu} \vspace{0.25in} \\ $^\dag$ Department of Computing and Mathematical Sciences \\ $^\ddag$ Department of Electrical Engineering \\ California Institute of Technology \\ Pasadena, CA 91125}

\maketitle

\begin{abstract}

%
%

Statistical estimation in many contemporary settings involves the acquisition, analysis, and aggregation of datasets from multiple sources, which can have significant differences in character and in value.  Due to these variations, the effectiveness of employing a given resource -- e.g., a sensing device or computing power -- for gathering or processing data from a particular source depends on the nature of that source.  As a result, the appropriate division and assignment of a collection of resources to a set of data sources can substantially impact the overall performance of an inferential strategy.  In this expository article, we adopt a general view of the notion of a resource and its effect on the quality of a data source, and we describe a framework for the allocation of a given set of resources to a collection of sources in order to optimize a specified metric of statistical efficiency.  We discuss several stylized examples involving inferential tasks such as parameter estimation and hypothesis testing based on heterogeneous data sources, in which optimal allocations can be computed either in closed form or via efficient numerical procedures based on convex optimization.


\end{abstract}

\textbf{Keywords}: Heterogeneous data sources, assignment problems, convex programming, resource tradeoffs in statistical estimation.

\section{Introduction} \label{SEC:intro}

Modern application domains throughout science and technology offer many opportunities for procuring and processing large amounts of data.  However, the effective deployment of resources for data acquisition and analysis is complicated by the fact that data are frequently obtained from multiple disparate sources, and the inferential objective involves an aggregation of these diverse datasets.  Specifically, different data sources typically have considerable variation in character and in value, and the effectiveness of a resource allotted to the treatment of a particular data source depends on the nature of the source.  Some examples of resources and their influence on the quality of a source include:
\begin{itemize}
\item \emph{Computing power}: algorithms employing more expensive processing and storage resources can improve the utility of a source.


\item \emph{Sensing devices}: in many scientific domains, using data acquisition devices more extensively or using more powerful sensors can provide data of better quality (e.g., larger datasets, data containing fewer errors).


\item \emph{Incentives for a population}: in settings involving surveys of a population, better incentives requiring a greater expenditure of resources on the part of the analyst can lead to higher quality data. For instance, participants may be more willing to provide informative answers (e.g., sacrifice some of their privacy) when given suitable inducements.
\end{itemize}
In each of these cases, the utilization of a resource involves a cost to the analyst.  Motivated by this observation, several researchers have investigated tradeoffs between the statistical accuracy of an inference algorithm and the amount of resources employed by the algorithm.  The tradeoff between statistical risk and computational resources has received a lot of attention \cite{DecGolRon98,Ser00,ChaJor13,ShaShaTom12,BerRig13b,WanBerSam14,MaWu13,Che13, FelGriRey13,FelPerVem13,ZhaDucWai13}, and those between risk and privacy constraints on estimation procedures have also been investigated recently \cite{AgrSri00,DucJorWai14}.

In this expository article, we study the \emph{optimal allocation} of resources in statistical estimation problems involving heterogeneous data sources.  In order to retain generality as well as broad applicability -- for example, to trade off and to allocate several types of resources simultaneously -- we adopt an abstract notion of a resource as a nondescript entity that is quantified by a real number.  Given $(i)$ functions that relate the quality of a data source to the amount of resource assigned to that source, and $(ii)$ a parameterized family of aggregation schemes (e.g., linear aggregators) for combining estimates obtained from multiple data sources, we design a joint strategy to allocate a set of resources to the different data sources and to aggregate estimates across the sources to optimize an overall metric of statistical efficiency.  From a technical as well as a conceptual point of view, our development differs from the literature on designing optimal methods for aggregating estimates from multiple data sources \cite{Bre96a,Bre96b,Wol92,BunTsyWeg07,Rig12,Bra13,BuhMei14}.  In particular, we consider only restricted families of (linear) aggregation schemes based on some knowledge about the distribution of the data, and the focus of our efforts is on the optimal allocation of resources to heterogeneous data sources.


\paragraph{Our stylized setup} Concretely, suppose there are $N$ independent heterogenous data sources, and in general terms, the source $i$ provides a random variable $\hat Y_i$ with loss $\ell_i \in \mathbb{R}$.  The loss is a measure of the imprecision associated to $\hat Y_i$, e.g., the variance of $\hat Y_i$, and it quantifies the accuracy of the source $i$.  The objective is to construct an aggregated estimator $\hat Y = a(\hat Y_1,\ldots, \hat Y_N)$ such that an overall loss $\Delta(a(\hat Y_1,\ldots,\hat Y_N);\ell_1,\ldots,\ell_N)$ is minimized:
\begin{equation*}
\min_{a \in \cA} \Delta(a(\hat Y_1,\ldots,\hat Y_N);\ell_1,\ldots,\ell_N).
\end{equation*}
Here $\cA$ is a constrained family of aggregation schemes.  Further, suppose that each of the losses $\ell_i$ is a function $\ell_i(r_i)$ of an amount $r_i \in \mathbb{R}$ of resource allocated to the source $i$; that is, the analyst utilizes the resource amount $r_i$ allotted to source $i$ and obtains in return a random variable $\hat Y_i$ from that source with loss $\ell_i(r_i)$.  As described above, the resources may be employed to acquire and/or process data, and the mapping $r_i \mapsto \ell_i(r_i)$ encodes the tradeoff between the quality of the source $i$ and the resource amount $r_i$ assigned to it.  In our abstraction, the analyst can only influence the quality of the source $i$ based on the resource amount $r_i$ allotted to the source (see Figure~\ref{FIG:scale} for an illustration).  Thus, in addition to choosing a suitable aggregator from the set $\cA$, the analyst also desires an allocation of resources to the $N$ sources to minimize the overall loss $\Delta(a(\hat Y_1,\ldots,\hat Y_N);\ell_1(r_1),\ldots,\ell_N(r_N))$:
\begin{equation}
\min_{r \in \cR} \min_{a \in \cA} \Delta(a(\hat Y_1,\ldots,\hat Y_N);\ell_1(r_1),\ldots,\ell_N(r_N)).  \label{EQ:generalra}
\end{equation}
\begin{figure}[h!]
    \begin{center}
    \begin{tabular}{cc}
      \includegraphics[width=.8\textwidth]{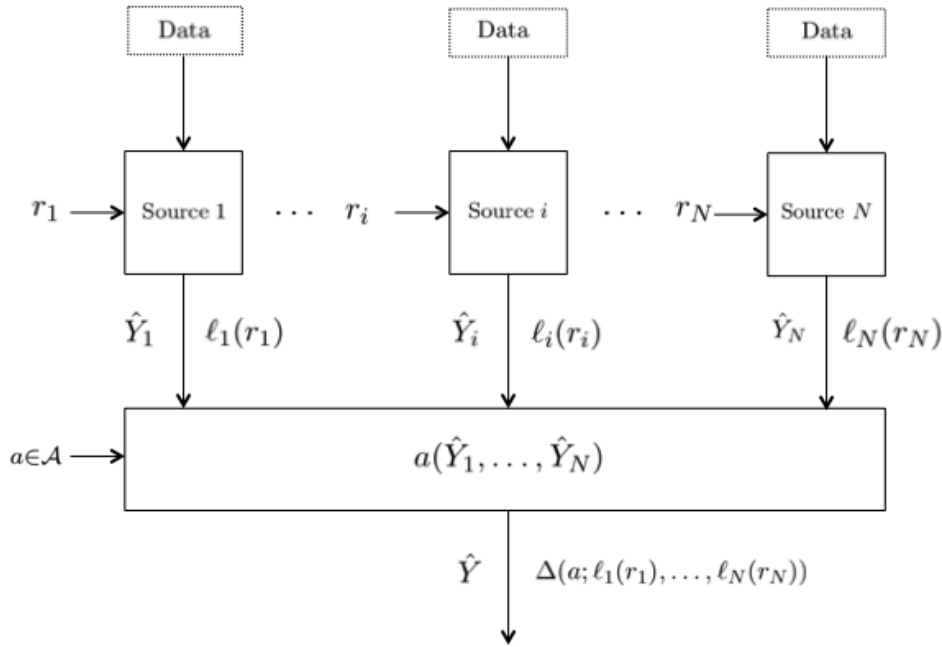} \\
    \end{tabular}
    \end{center}
    \vskip -2ex
    \caption{Illustration of our setup: the analyst chooses an allocation of resources $(r_1,\ldots,r_N) \in \mathcal{R}$ to the different sources and an aggregation scheme $a \in \mathcal{A}$ to obtain an estimate $\hat Y$ that minimizes the overall loss $\Delta(a; \ell_1(r_1),\ldots,\ell_N(r_N))$.}
    \label{FIG:scale}
\end{figure}

Here $\cR \subset \mathbb{R}^N$ encodes the constraints facing the analyst on the manner in which resources may be allotted to the different sources.  We assume that the overall loss function $\Delta$, the set of aggregators $\cA$, and the resource constraint set $\cR$ are specified by the analyst, and that the resource-quality tradeoff functions $\ell_i(r_i), ~ i=1,\ldots,N$ are known in advance.  The prior knowledge of these tradeoff functions may be a somewhat restrictive assumption in some cases; however, in many settings in which an inferential task is to be performed on a regular basis, the intrinsic qualities of different data sources and their dependence on various resources are feasible to estimate from past observations (e.g., financial asset modeling, marketing based on online surveys).  In such situations, the optimization problem \eqref{EQ:generalra} provides an optimal allocation of resources to minimize the overall loss $\Delta$.

As specific instances of the general setup outlined here, we describe two canonical settings.  In Sections \ref{SEC:vectestim} and \ref{SEC:generallinearparam} we discuss the problem of estimating an unknown parameter in $\mathbb{R}^d$ in which each of the $N$ sources provides information about the parameter in the form of a linear image of the parameter corrupted by Gaussian noise.  The two causes of heterogeneity in this case are the variations in the linear maps as well as the noise variances across the different sources.  We investigate the problem of optimal resource allocation when the noise variance of each source is influenced by the resource amount allocated to the source (as specified by a known tradeoff function).  Depending on the nature of the source, we demonstrate that in some cases it is better to allocate more resources to lower quality sources, while in others it is preferable to allocate more resources to higher quality sources.  In our next setting in Section~\ref{SEC:halfspace}, we study the optimal allocation of resources in a hypothesis testing task in which the objective is to determine whether an unknown parameter lies on a specified side of a given hyperplane.  Each source provides information about one of the coordinates of the parameter, with the precision of the estimate being dependent on the resource allocated to the source.  We consider cases in which the unknown parameter lies in the unit closed hypercube $[0,1]^d$ and in the set $\{0,1\}^d$; see Section~\ref{SEC:halfspace} for further details.

In these examples, we consider two types of resource constraint sets.  Perhaps the most elementary example of a resource constraint set $\mathcal{R}$ is one specified by a simplex:
\begin{equation*}
\cU_N = \{r \in \mathbb{R}^N: r \succeq 0, ~ r_1+\ldots+r_N \leq 1\}.
\end{equation*}
Such a resource constraint set corresponds to the situation in which resources are infinitely divisible.  A second type of constraint set that we consider is one in which there are $N$ possible resources with fixed resource amounts $r_1,\ldots,r_N$, and each resource can be assigned to exactly one of the $N$ data sources.  Such types of constraints are relevant if the resources are physical devices that are used to acquire or process data.  For each of these constraint sets, we describe conditions under which the optimal allocation of resources \eqref{EQ:generalra} can be computed efficiently.  We discuss cases in which the optimal allocation can in fact be obtained in closed form, as well as several others in which the optimal allocation can be computed numerically in a tractable manner via convex optimization.

\subsection{Related work} \label{SUBSEC:relatedwork}
Resource allocation is a prominent subfield of operations research, with an emphasis on computationally tractable techniques for obtaining optimal allocations in problem domains such as supply chain, logistics, and transportation.  In contrast to the applications considered in that literature, our emphasis is on the development of resource allocation strategies for statistical inference problems.  In the information sciences, a prominent example of a resource allocation problem is that of allocating power across a collection of independent communication channels of varying capacities for the purpose of maximizing overall throughput \cite{Gal68,CovTho91}.  In this case, the optimal allocation of power to the different channels is given by the famous water-filling formula \cite{Hol64}.  In the area of sensor resource management, the problem of optimal sensor placement can also be viewed from the perspective of resource allocation \cite{HerCasCoc07}.

Our setup is different from that of bandit problems in online learning, in which the quality / performance of each ``arm'' of a bandit (in our case, the sources) is unknown and the processing / aggregation is done in an online fashion as the data are acquired (see \cite{BubCes12} for more information on this problem, which was first studied by \cite{Tho33}). In comparison, in our setting the quality of a source as a function of resources allocated to the source is assumed to be known in advance, and the resource allocation optimization problem \eqref{EQ:generalra} is solved offline before any data are acquired or analyzed.

\subsection{Notation} \label{SUBSEC:notation}
For a positive integer $d$, the set $\{1,\ldots,d\}$ is denoted $[d]$.  The cardinality of a subset $S \subset [d]$ is denoted by $|S|$.  For $x \in \mathbb{R}^d$, we denote the $i$'th coefficient of the vector by $x_i$; the subvector of $x$ with coordinates corresponding to a subset $S \subset [d]$ is denoted by $x_S \in \mathbb{R}^{|S|}$.  For a collection $v_1,\ldots,v_n$ of vectors of $\mathbb{R}^d$, the $j$'th coefficient of $v_i$ is denoted by $v_i^{(j)}$ to avoid ambiguity.  For $u ,v\in \mathbb{R}^d$, we denote by $\langle u,v \rangle$ the Euclidean scalar product of $\mathbb{R}^d$ and by $\|u\|_2 = \sqrt{\langle u,u \rangle}$ the associated Euclidean norm of $u$. We denote by $\cU_d$ the unit simplex of $\mathbb{R}^d$, and by $\mathfrak{S}_n$ the symmetric group (the set of permutations of $n$ elements). When the choice of distribution is clear, the notations $\bP$ and $\E$ refer to the probability and expectation relative to that distribution.

\section{A Preliminary Example of Parameter Estimation from Heterogeneous Sources}
\label{SEC:vectestim}


\subsection{Problem description} \label{SUBSEC:vectestimintro}
We consider the problem of estimating a parameter $\theta \in \mathbb{R}^d$ based on $N$ independent data sources.  The sources provide independent random variables $\hat \theta_1,\ldots, \hat \theta_N$, each with mean $\theta$, and the losses $\ell_i$ corresponding to these sources are the mean squared errors of $\hat \theta_1,\ldots, \hat \theta_N$.  Thus, for each $i \in [d]$, allocating resource $r_i \ge 0$ to the data source $i$ yields an estimator $\hat \theta_i$ with mean squared errors
$$\E[\|\hat \theta_i - \theta \|_2^2] = \ell_{i}(r_i)\, ,$$
where $\ell_i$ is a positive, decreasing function. We consider the case in which the variances of the coefficients of $\hat \theta_i$ are identical.


We combine the estimators $\hat \theta_1,\ldots, \hat \theta_N$ by a linear aggregation scheme as follows:
\begin{equation}
\label{EQ:thetalambda}
\hat \theta_\lambda = a_\lambda(\hat \theta_1,\ldots, \hat \theta_N)
= \sum_{i=1}^N \lambda_i \hat \theta_i\, ,
\end{equation}
for $\lambda \in \cU_N$, i.e., the set $\cA$ of aggregations is given by the collection of convex combinations of the estimators $\hat \theta_1,\ldots,\hat\theta_N$.  As each of the estimators $\hat \theta_1,\ldots,\hat\theta_N$ is unbiased, we have that $\E[\hat \theta_\lambda] = \theta$.  Further, letting the overall loss $\Delta$ be the mean squared error of the estimator $\hat \theta_\lambda$, the independence of the data sources implies that:
$$\Delta(a_\lambda(\hat \theta_1,\ldots, \hat \theta_N); \ell_{1}(r_1),\ldots,\ell_{N}(r_N)) = \E[\|\hat \theta_\lambda - \theta \|_2^2] = \sum_{i=1}^N \lambda_i^2 \ell_{i}(r_i) \, .$$
Our objective is therefore to optimize both the allocation of resources (the variable $r$ in a resource constraint $\cR$) and the aggregation of the estimators (the variable $\lambda \in \cU_N$) in order to minimize the overall loss $\Delta$:
\begin{equation}
\label{EQ:optimglobal1}
\begin{aligned}
& \text{minimize} &&  \sum_{i=1}^N \lambda_i^2 \ell_{i}(r_i) &\\
& \text{subject to}&& \lambda \in \cU_N&\\
& && r \in \cR \\
& \text{variables}&& \lambda , r \in \mathbb{R}^N\, .&\\
\end{aligned}
\end{equation}
This optimization problem can be simplified as follows:

\begin{proposition}
\label{PRO:optimreduced}
For positive loss functions $\ell_i$, the optimization problem \eqref{EQ:optimglobal1} can be reformulated as
\begin{equation}
\begin{aligned}
\label{EQ:simplex}
& \text{minimize} &&  1 / \sum_{i=1}^N \ell^{-1}_i(r_i) &\\
& \text{subject to}&&  r \in \cR. \\
\end{aligned}
\end{equation}
\end{proposition}
\begin{proof}
In order to obtain this formulation we fix $r \in \mathcal{R}$ in \eqref{EQ:optimglobal1}, so that $\ell_i = \ell_i(r_i)$ is also fixed, and we optimize over $\lambda \in \cU_n$:
\begin{equation}
\label{EQ:optimsimplex}
\begin{aligned}
& \text{minimize} && \sum_{i=1}^N \lambda_i^2 \ell_{i}&\\
& \text{subject to}&& \lambda \in \cU_N. &\\
\end{aligned}
\end{equation}
The optimization problem \eqref{EQ:optimsimplex} projects the origin onto the unit simplex according to a reweighted $\ell_2$ norm.  One can check that the optimal solution is $\lambda_i^* = \ell_i^{-1} / \sum_{i'=1}^N \ell_{i'}^{-1}, ~ i=1,\ldots,N$, which corresponds to the optimal value $\E[\|\hat \theta_{\lambda^*} - \theta \|_2^2] =   1 / \sum_{i=1}^N \ell_i^{-1}$.
\end{proof}

The aggregate estimate based on these optimal weights is given by $\hat \theta_{\lambda^*} = \sum_{i=1}^N \ell_i^{-1} \, \hat \theta_i / (\sum_{i=1}^N \ell_i^{-1})$.  We note that the naive choice of $\lambda_i = 1/N$ would yield an overall loss of $\sum_{i=1}^N \ell_i/ N^2$, which is always bounded below by $1 / \sum_{i=1}^N \ell_i^{-1}$ based on the arithmetic-geometric-mean inequality.  It is worthwhile to notice (as an added justification of the optimality of this aggregated estimator) that in the case $\hat \theta_i = \cN(0,\ell_i)$, the estimator $\hat \theta_{\lambda^*}$ is the maximum-likelihood estimator of $\theta$ for known $\ell_i$'s.\\



It is sometimes more convenient to parameterize the tradeoff function $\ell_i$ via its inverse:
\begin{equation*}
q_i(r_i) = \ell^{-1}_i(r_i),
\end{equation*}
which can be viewed as the precision of the estimator $\hat \theta_i$.  Since the loss functions $\ell_i(r_i)$ are assumed to be positive, the precision functions $q_i(r_i)$ are also positive.  Consequently, with respect to this alternative parameterization and based on Proposition~\ref{PRO:optimreduced}, the optimization problem \eqref{EQ:optimglobal1} can be simplified as:
\begin{equation*}
\begin{aligned}
& \text{maximize} && \sum_{i=1}^N q_i(r_i) &\\
& \text{subject to}&&  r \in \cR .&\\
\end{aligned}
\end{equation*}
Next, we consider this optimization problem for two choices for the constraint set $\cR$.\\

\subsection{Simplex constraint} \label{SUBSEC:vectestimsimplex}

The simplest case of a resource constraint set $\mathcal{R}$ is one in which the resources are infinitely divisible and the total resource budget is $R >0$:
\begin{equation}
\begin{aligned}
\label{EQ:simplex}
& \text{maximize} &&  \sum_{i=1}^N q_{i}(r_i) &\\
& \text{subject to}&&  r \in R \cdot \cU_N &\\
\end{aligned}
\end{equation}
Here $R  \cdot \cU_N = \{r \in \mathbb{R}^N : r \succeq 0, ~ \sum_{i=1}^N r_i \leq R\}$ .  If the precision functions $q_i(r_i)$ are concave and tractable to compute, then the optimization problem \eqref{EQ:simplex} is a convex program that can be solved efficiently.  Indeed, the $q_i$'s being concave and non-decreasing corresponds to the case in which additional resources improve the quality of a source but with a ``diminishing returns'' effect, a situation that is quite natural in many settings.  We note that $q_i$ being positive, concave, and non-decreasing implies that $\ell_i$ is positive, non-increasing, and convex.


Perhaps the most natural example of a resource-loss tradeoff function is $\ell_i(r_i) = \sigma_i^2/r_i$ or $q_i(r_i) =r_i/\sigma_i^2$, where $\sigma^2_i$ may be viewed as the ``intrinsic'' error variance of each component of source $i$. In this case, allocating $r_i$ to data source $i$ may be viewed equivalently as sampling from source $i$ ``$r_i$ times.'' If $\sigma_1 \le \ldots \le \sigma_N$, the optimal solution of \eqref{EQ:simplex} is $r^* = (R,0,\ldots,0)$; that is, the optimal strategy is to allocate all the resources to the best data source, i.e., the one with the smallest intrinsic variance. The optimal aggregation is also to focus entirely on one source, and $\hat \theta_{\lambda^*} = a_{\lambda^*} = \hat \theta_1$.  The interpretation of this result is that it is optimal to simply sample from the best source.

This effect is mitigated if $\ell_i(r_i) = \sigma^2/r_i^\alpha$ for $0<\alpha<1$.  For such loss functions, the KKT conditions of \eqref{EQ:simplex} yield the following optimal solution:
$$r_i^* = R  \frac{(\sigma_i^2/\alpha)^{\frac{1}{\alpha-1}}}{\sum_{j=1}^n(\sigma_j^2/\alpha)^{\frac{1}{\alpha-1}}} \, .$$
Again the better sources receive a greater fraction of the allocated resource, although the best source does not exclusively receive all the resources.  When $\alpha\rightarrow 1$, this solution converges to the extreme case above of the optimal solution for $\alpha=1$ (all resources going to the best source).

Another interesting example of a precision function for which there is a closed-form solution with an illuminating interpretation is
$$q_i(r_i) = \frac{1}{\sigma^2_i} + \log\Big(1+\frac{r_i/R}{a_i}\Big)\, .$$
This setting models a situation in which each variance is initially $\sigma_i^2$, and where any positive resource $r_i>0$ allocated to a source improves the precision at a rate given by $a_i$ in a concave manner (independently of the initial variance). Minimizing the expected loss in this case is mathematically equivalent to maximizing the communication rate over $N$ channels by allocating power $r_i/R$ to the $i$'th transmitter (see \cite{Hol64,Gal68,CovTho91}). The solution to this problem is given by the well-known {\em water-filling} method
$$r_i^* = R \, \max\{0,A - a_i\} \, ,$$
where $A$ is chosen such that the $r_i^*$ sum to $R$. Here the optimal allocation strategy is blind to the initial quality of each source (i.e., not influenced by $\sigma^2_i$), but is based on the possible improvements realized by allocating resources. Assuming that the $a_i$'s are different, the resources are initially allocated to the source with lowest $a_i$ as that source is the one in which the initial marginal improvement is highest.  Once this improvement decreases to the level of the second highest marginal improvement, the resources are subsequently divided equally between these two sources, and so on.  This process is repeated until all resources are exhausted, which is the source of the name of this method.

These are just a few simple cases of resource allocation problems with closed-form solutions.  Finally, we note that for any concave precision functions $q_i$ (consistent with convex loss functions $\ell_i$), adding further convex inequalities to the resource constraint set $\mathcal{R}$ in the problem \eqref{EQ:simplex} still yields a convex program; these in turn can also be solved efficiently.  Our setup is therefore adaptable to further limitations on the allocated resources that can be expressed as convex constraints on $r$ (e.g., bound on the maximal or minimal amount allocated to each source, on the concentration of resources on a few sources). \\



\subsection{Assignment constraint} \label{SUBSEC:vectestimassign}

A qualitatively different type of constraint on the allocation to the setting above is the situation in which there are $N$ possible resources with fixed values $r_1,\ldots,r_N$, and each resource is assigned to exactly one data source. This would, for example, be the case for physical devices that acquire or process the data. In this setting, the optimization problem \eqref{EQ:optimglobal1} (via the reformulation \eqref{EQ:simplex}) becomes
\begin{equation}
\begin{aligned}
\label{EQ:assignment}
& \text{maximize} &&  \sum_{i=1}^N q_{i}(r_{\tau(i)}) &\\
& \text{subject to}&&  \tau \in \mathfrak{S}_N \, ,
\end{aligned}
\end{equation}
where $\mathfrak{S}_N$ corresponds to the set of permutations on $N$ elements.  This problem is known as the \emph{assignment} problem (and it is also a special case of the \emph{optimal transport} problem), and it can be solved efficiently using several methods, e.g., by linear programming or by the Hungarian algorithm \cite{Kuh55}.  In the linear programming approach, one considers the convex hull of the set of $N \times N$ permutation matrices, which gives an equivalent optimization problem to \eqref{EQ:assignment} in terms of the Birkhoff polytope $B_N$ of doubly stochastic matrices.  By taking $Q_{ij} = q_i(r_j)$, the problem \eqref{EQ:assignment} can be reformulated as
\begin{equation}
\begin{aligned}
\label{EQ:birkhoff}
& \text{maximize} &&  \tr(QM) &\\
& \text{subject to}&&  M \in B_N \, .
\end{aligned}
\end{equation}
There exists an optimal solution $M^*$ that is a permutation matrix, which specifies the optimal resource assignment.  This problem can be solved efficiently using standard solvers for linear programming.  One can also obtain closed-form solutions for special cases of $Q$.  For example, consider again the situation in which $\ell_{ij} = \sigma_i^2/\, r_j$, or $q_{ij} = r_j / \sigma_i^2$.  Assuming that the data sources as well as the resources are ranked by quality, i.e., $r_1 \ge \ldots \ge r_n$ and $\sigma_1^2 \le \ldots \le \sigma_n^2$, the matrix $Q$ has rank one and an optimal assignment is $\tau^*(i)=i$ due to the reordering inequality.  This problem can also be interpreted as a probabilistic version of the optimal transport problem. Suppose $(X,Y)$ are random variables with marginal distributions uniform on $\{r_1,\ldots,r_n\}$ and $\{1/\sigma^2_1,\ldots,1/\sigma^2_n\}$ respectively.  Finding the joint distribution that minimizes the expected cost $\mathbb{E}_{X,Y}[(x-y)^2]$ is equivalent to the optimization problem \eqref{EQ:birkhoff}.


As in Section~\ref{SUBSEC:vectestimsimplex}, it is again the case that better quality sources should be favored both in the choice of $\lambda$ and $r$.  More generally, the situation is the same for $Q_{ij} = \phi(r_j)/\, \sigma_i^2$ for any increasing function $\phi$. The function $\phi(r)=r^\alpha$ that we discussed for the simplex constraint is a special case of this more general class of functions.
\\

\section{Parameter Estimation from Linear Measurements} \label{SEC:generallinearparam}
We consider two successive generalizations of the linear parameter estimation problem of Section~\ref{SEC:vectestim}. We first study the setting in which each data source provides information about an arbitrary subset of the coefficients of $\theta \in \mathbb{R}^d$. We then generalize that problem further by investigating the case in which each data source provides an estimate of an arbitrary linear function of $\theta \in \mathbb{R}^d$.
\subsection{Sources with heterogenous supports} \label{SUBSEC:generalsupport}
\label{SEC:gaussmeansets}


The setting described in the previous section is a simple illustration of a more general class of problems that we consider next.  Source $i$ provides an estimate $\hat \theta_i \in \mathbb{R}^{|S_i|}$ of the vector $\theta_{S_i}$ corresponding to a subset $S_i \subseteq [d]$. One example is the case in which the $i$'th data source provides an estimate of the $i$'th coefficient of $\theta$.  In this situation there are $N = d$ sources and $S_i = \{i\}$. In the previous section, each $S_i = [d]$ is equal to the whole parameter set. Heterogeneity among the sources can manifest itself in terms of different loss functions $\ell_i(r_i)$ (e.g., the sources have different intrinsic variances, as in Section~\ref{SEC:vectestim}), in the set of coefficients estimated by each source, and in the different variances among coefficients of a given $\hat \theta_i$.  As before, we assume that the variable $\hat \theta_i$ has mean $\E[\hat \theta_i] = \theta_{S_i}$ for a given {\em observation set} $S_i \subset [d]$, and we have
$$\E[(\hat \theta_i^{(j)} - \theta^{(j)}_{S_i})^2] = \ell^{(j)}_i(r_i) \, .$$
That is, the variance of each component of $\hat \theta_i$ could be different, and is explicitly characterized as a function of the resource $r_i$.  Following the development in Section~\ref{SUBSEC:vectestimintro}, we consider first the optimal aggregation problem with $r$ fixed.  Let $\hat \Theta \in \mathbb{R}^{d \times N}$ be a matrix with columns $\hat \theta_1,\ldots,\hat \theta_N$ (these estimators are extended to $\mathbb{R}^d$ by appropriate zero-padding). For each $\Lambda \in \mathbb{R}^{d \times N}$, we consider the aggregated estimator
$$\hat \theta_\Lambda = a_\Lambda(\hat \theta_1,\ldots \hat \theta_N) = \text{diag}(\hat \Theta \Lambda^T)\, .$$
For each $j \in [d]$, let $I_j = \{i : j \in S_i\} \subseteq [N]$ be the $j$'th {\em reciprocal set} of the observation sets. We then have that the $j$'th coordinate $\hat \theta^{(j)}_\Lambda$ of $\hat \theta_\Lambda$ is described in terms of the $j$'th row $\hat \Theta^{(j)} \in \mathbb{R}^N$ of $\hat \Theta$ and the $j$'th row $\Lambda^{(j)} \in \mathbb{R}^N$ of $\Lambda$ as follows:
$$\hat \theta_\Lambda^{(j)} = \sum_{i=1}^N \Lambda_i^{(j)} \hat \Theta_i^{(j)} =  \sum_{i \in I_j} \Lambda_i^{(j)} \hat \Theta_i^{(j)} \, .$$
As before, we constrain our collection of aggregation schemes to suitable convex combinations of the estimates $\hat \theta_i$ via the following restriction on $\Lambda$: For each $j \in [d]$, the $j$'th row $\Lambda^{(j)} \in \mathbb{R}^N$ of $\Lambda$ satisfies the constraint that $\Lambda^{(j)} \in \cU_N$. We wish to minimize the overall loss

$$\Delta(a_\lambda(\hat\theta_1,\ldots,\hat\theta_N); \ell_{1}(r_1),\ldots,\ell_{N}(r_N)) = \E[\|\hat \theta_\Lambda - \theta \|_2^2] \, .$$
This yields the following optimization problem
\begin{equation}
\label{EQ:optimglobal}
\begin{aligned}
& \text{minimize} &&  \Delta(a_\lambda(\hat\theta_1,\ldots,\hat\theta_N); \ell_{1}(r_1),\ldots,\ell_{N}(r_N)) &\\
& \text{subject to}&& \Lambda^{(j)} \in \cU_N, ~ \forall \, j \in [d]&\\
& && r \in \cR \\
& \text{variables}&&  \Lambda \in \mathbb{R}^{N \times d}, r \in \mathbb{R}^N \, .
\end{aligned}
\end{equation}
By following the same line of reasoning described in Section~\ref{SEC:vectestim} (essentially the problem \eqref{EQ:optimglobal} is equivalent to $d$ parallel one-dimensional problems of the type considered in Section~\ref{SEC:vectestim}), the optimization over $\Lambda$ (with $r \in \mathbb{R}^N$ fixed) yields
\begin{equation}
\label{EQ:lambdastarcap}
\begin{aligned}
{\Lambda_i^{(j)}}^* &= {\ell_i^{(j)}}^{-1} / \sum_{i' \in I_j} {\ell_{i'}^{(j)}}^{-1} ~\text{for}~ i\in I_j \\ \E[(\hat \theta^{(j)}_{\Lambda^*} - \theta^{(j)} )^2] &= 1 / \sum_{i \in I_j} {\ell^{(j)}_i}^{-1} \\ \E[\|\hat \theta_{\Lambda^*} - \theta \|_2^2] &= \sum_{j=1}^d \frac{1} { \sum_{i \in I_j} {\ell^{(j)}_i}^{-1}}.
\end{aligned}
\end{equation}
Letting $q_i^{(j)} = {\ell_i^{(j)}}^{-1}$ and using the fact that the losses $\ell^{(j)}_i$ are positive, we have that \eqref{EQ:optimglobal} can be simplified as
\begin{equation}
\begin{aligned}
\label{EQ:setsoptim}
& \text{minimize} && \sum_{j=1}^d \frac{1} { \sum_{i \in I_j} q^{(j)}_i(r_i) } &\\
& \text{subject to}&&  r \in \cR \, .&
\end{aligned}
\end{equation}

The situation appears more complicated than in the previous case in Section~\ref{SEC:vectestim}, but this problem is still tractable to solve numerically under suitable conditions:
\begin{proposition}
Suppose each $q_i^{(j)}$ is a concave, non-decreasing, and positive function. Then the objective function $\sum_{j=1}^d \frac{1} { \sum_{i \in I_j} q^{(j)}_i(r_i) }$ of \eqref{EQ:setsoptim} is convex.
\end{proposition}
\begin{proof}
We use well-known rules of composition \cite{BoyVan04}. The functions
$\frac{1} { \sum_{i \in I_j} q^{(j)}_i(r_i) }$ are convex, as the function $y \mapsto 1/y$ is non-increasing and convex on $\mathbb{R}_+$, and the sum of the $q^{(j)}_i$'s is concave (as they are individually concave). The objective function $\sum_{j=1}^d \frac{1} { \sum_{i \in I_j} q^{(j)}_i(r_i) }$ is therefore convex, as it is a sum of convex functions.
\end{proof}

Thus, for any choice of $S_1,\ldots,S_N$, this problem can be numerically solved as a convex optimization problem.  In the following two examples corresponding to extreme cases of total redundancy (where the $S_i$ are all the same, equal to $[d]$) and total independence (where the $S_i$ are all disjoint, such as when $S_i = \{i\}$), we demonstrate the richness of this general setting.  In particular these examples illustrate that different types of support sets can substantially alter the optimal resource allocation strategies.\\

\paragraph{Total redundancy} Here each $S_i = [d]$ and we recover the example studied in Section~\ref{SEC:vectestim}.  Specifically, in \eqref{EQ:setsoptim} we set $\ell_i^{(j)} = \ell_i /d$ and $S_i=[d]$ (and hence $I_j=[N]$) for all $i \in [N]$ and $j \in [d]$.  For the precision functions $q^{(j)}_i(r_i) =r_i/\sigma_i^2$, the optimal strategy is to allocate all the resources to the best source (the one with smallest intrinsic variance $\sigma_i^2$), as discussed in Section~\ref{SEC:vectestim}.\\

\paragraph{Total independence}  In the other extreme, if all the sets $S_i$ are disjoint, we can assume without loss of generality that $S_i = \{i\}$ and $N=d$ (each $I_j$ is a singleton set) and the aggregation weights are $\Lambda_i^{(j)} = 1$ for $i \in I_j$ and $0$ otherwise.  We have from \eqref{EQ:setsoptim} that the optimal resource allocation is obtained by solving
\begin{equation}
\begin{aligned}
& \text{minimize} && \sum_{j=1}^d \frac{1} { q_i(r_i) } &\\
& \text{subject to}&&  r \in \cR \, .
\end{aligned}
\end{equation}
The contrast with the first extreme case of total redundancy can be made very apparent by again taking $q_i(r_i) =r_i/\sigma_i^2$.  When the constraint set is $R  \cdot \cU_N$, the KKT conditions yield the optimal strategy
$$r_i^* = R \frac{\sigma_i}{\sum_{i=j}^N \sigma_j}\, .$$
In this case, more of the resources are allocated to the \emph{lower-quality} sources. Unlike in the previous example, there is only one source of information for each coefficient of $\theta$.  Therefore, a single ``weak source'' can affect the overall performance of the inference procedure, and as a result the sources with greater variance (i.e., the weaker ones) receive priority in resource allocation in order to obtain the highest quality inferential outcome.  A similar reasoning also holds if the resource constraint set is changed to an assignment type constraint.  Suppose we have $N$ resources with fixed resource values $r_1,\ldots,r_N$; the analog of the optimization problem described in \eqref{EQ:assignment} is
\begin{equation}
\begin{aligned}
& \text{minimize} &&  \sum_{i=1}^N 1/q_{i}(r_{\tau(i)}) &\\
& \text{subject to}&&  \tau \in \mathfrak{S}_N \, .
\end{aligned}
\end{equation}
The optimal strategy is again the opposite of the one established in the total redundancy setting. If $\sigma_1 \le \ldots \sigma_N$ and $r_1 \ge \ldots \ge r_N$, the optimal assignment is $\tau^*(i) = n-i+1$ by using the reordering inequality on the inverse $\sigma_i^2/r_{\tau(j)}$.


\subsection{A numerical example} \label{SUBSEC:supportexample}
We consider optimal resource allocation in a statistical estimation task based on the setting described in Section~\ref{SUBSEC:generalsupport}, with dimension $d=10$ and $N=5$ sources.  These sources provide estimates of subsets of coordinates as described in the following table:
\begin{center}
\begin{tabular}{ |l|r| }
  \hline
  \multicolumn{2}{|c|}{Observation subsets} \\
  \hline
  $S_1$ & $\{3,5,7,10\}$ \\
  $S_2$ & $\{5,8,10\}$ \\
  $S_3$ & $\{2,7\}$ \\
  $S_4$ & $\{1,2,4,6,7,9\}$ \\
  $S_5$ & $\{3,4,7\}$\\
  \hline
\end{tabular}
\quad \quad \quad \quad
\begin{tabular}{ |l|r| }
  \hline
  \multicolumn{2}{|c|}{Reciprocal sets} \\
  \hline
  $I_1$ & $\{4\}$ \\
  $I_2$ & $\{3,4\}$ \\
  $I_3$ & $\{1,5\}$ \\
  $I_4$ & $\{4,5\}$ \\
  $I_5$ & $\{1,2\}$\\
  $I_6$ & $\{4\}$ \\
  $I_7$ & $\{1,3,4,5\}$ \\
  $I_8$ & $\{2\}$ \\
  $I_9$ & $\{4\}$ \\
  $I_{10}$ & $\{1,2\}$\\
  \hline
\end{tabular}
\end{center}
\begin{itemize}
\item {\em Linear precision}: If the precision functions are specified as $q_i^{(j)} = r_i/\sigma_i^2$, with $\sigma_i^2 = i$, then the optimal allocation of resources in $\cU_5$ that we obtain by solving the convex program \eqref{EQ:setsoptim} is
$$r^* \approx (0.194,0.207,0.00,0.599,0.00) \, .$$
This solution highlights the fact that precision functions of the form $q_i^{(j)} = r_i/\sigma_i^2$ (corresponding, for example, to the number of times a source is sampled) yield sparse optimal allocation strategies.  As a matter of fact, it is clear that we should have $r^*_5=0$: sources 1 and 4 both have a lower noise variance for the coordinates that appear in subset $S_5$ (i.e., $S_5 \subset S_4 \cup S_1$ and $\sigma_5 > \sigma_4>\sigma_1$).
\item {\em Power precision}: With the same setup as above, but the precision functions now given as $q_i^{(j)} = r_i^\alpha/\sigma_i^2$ (with $\sigma_i^2 = i$ and $\alpha=0.6$), the optimal allocation of resources is:
$$r^* \approx (0.160,0.177,0.018,0.631,0.014) \, .$$
Unlike the linear precision case, the optimal solution is not sparse anymore as the power $\alpha <1$; rather, the resource amounts allocated to some of the sources are quite small instead of being equal to $0$.


\end{itemize}

\subsection{Parameter estimation from general linear measurements} \label{SUBSEC:generallinearparam}
In this section, we take a somewhat more general viewpoint of estimating an unknown parameter $\theta \in \mathbb{R}^d$ from linear measurements than those described in the preceding discussions.  We associate to the $i$'th data source a linear functional specified by a vector $X^{(i)} \in \mathbb{R^d}$, and the source provides the following random variable:
\begin{equation}
\label{EQ:reglin}
y_i = \langle X^{(i)} , \theta \rangle +\varepsilon_i \,.
\end{equation}
The noise vector $\varepsilon \sim \mathcal{N}(0,P^{-1})$ is Gaussian, where $P \in \mathbb{S}^N_+$ is a positive definite precision matrix.  Letting $X \in \mathbb{R}^{N \times d}$ be a matrix with the $i$'th row being equal to $X^{(i)}$, we have that
$$y = X \theta + \epsilon \, .$$
We assume that $X$ is full rank, which implies in particular that $N \geq d$.  We consider the following aggregation of the components of $y$ to obtain the minimum-variance unbiased estimator of $\theta$:
\begin{equation}
\hat \theta =  a(y_1,\ldots,y_N) = (X^\top P X)^{-1} X^\top y. \label{EQ:mvue}
\end{equation}
As $\hat \theta - \theta^* \sim \cN(0,(X^\top P X)^{-1})$, the mean squared error of this estimator is given by
$$\E [\| \hat \theta - \theta^* \|^2_2] = \tr((X^\top P X)^{-1}) \,.$$

We parameterize resources by the precision matrix $P$, so that restrictions on the manner in which resources are allocated are specified via a constraint set $\mathcal{P} \subset \mathbb{S}^N_+$.  This is a generalization of the problems considered in Sections~\ref{SEC:vectestim} and \ref{SUBSEC:generalsupport}.  For example, suppose $X$ is composed of $N$ rectangular blocks of rows, such that the $i$'th block (corresponding to the $i$'th data source) consists of $|S_i|$ rows and the $j$'th row of the $i$'th block is $e_{S_i(j)}$.  Let $P$ be a diagonal matrix composed of $N$ segments such that the $i$-th segment has cardinality $|S_i|$, and where the $j$'th element of the $i$'th segment is $q_i^{(j)}(r_i)$.  With this choice of $X$ and $P$, we clearly recover the problem described in Section~\ref{SUBSEC:generalsupport}.

The function that maps $P \in \mathbb{S}^N_+$ to $\tr((X^\top P \, X)^{-1})$ can be shown to be convex based on standard composition rules, thus yielding the following result:
\begin{proposition}
If the resource constraint $\mathcal{P} \subset \mathbb{S}^N_+$ is a convex set, then minimizing the mean squared error $\E [\| \hat \theta - \theta \|^2_2]$ is equivalent to solving the following convex optimization problem:
\begin{equation}
\begin{aligned}
\label{EQ:setsoptim1}
& \text{minimize} && \tr((X^\top P \, X)^{-1}) &\\
& \text{subject to}&&  P \in \cP \, .
\end{aligned}
\end{equation}
\end{proposition}
\begin{proof}
The map $M \mapsto \tr(M^{-1})$ is convex over the set $\mathbb{S}^N_+$ \cite{LewSen05}, and the map $P \mapsto X^\top P X$ is a linear map.  Consequently, the map $P \mapsto \tr((X^\top P \, X)^{-1})$ is convex.
\end{proof}

Other measures of the performance of the estimator \eqref{EQ:mvue} may also be of interest.  For instance, if the focus of the user is on a high-probability guarantee on the deviation $\| \hat \theta - \theta^* \|^2_2$, rather than a guarantee in expectation, it is possible to modify our approach accordingly. Indeed, we have the following upper bound as a consequence of \cite[Lemma 1]{LauMas00}.
$$\bP[\| \hat \theta - \theta \|^2_2 > 2 \|(X^\top P \, X)^{-1}\|_F \sqrt{t} + 2\|(X^\top P \, X)^{-1}\|_{\text{op}} t ] \le e^{-t} \, ,$$
where $\hat \theta$ is as defined in \eqref{EQ:mvue}.  Therefore, in order to find an $\ell_2$ ball with a minimal upper bound on the radius and with confidence $1-\delta$ , one can solve the following optimization problem with $\lambda=\sqrt{\log(1/\delta)}$
\begin{equation}
\begin{aligned}
\label{EQ:setsoptim2}
& \text{minimize} && \|(X^\top P \, X)^{-1}\|_F + \lambda \|(X^\top P \, X)^{-1}\|_{\text{op}} &\\
& \text{subject to}&&  P \in \cP \, .
\end{aligned}
\end{equation}
This problem is also convex (if $\mathcal{P} \subset \mathbb{S}^N_+$ is a convex set) as the map $P \mapsto \|(X^\top P \, X)^{-1}\|$ is convex for any unitarily-invariant matrix norm $\|\cdot\|$ \cite{LewSen05}.

The optimization problems \eqref{EQ:setsoptim1} and \eqref{EQ:setsoptim2} can also be solved as convex programs if $P$ is fixed, and the optimization is over $X$, i.e., the resource allocation problem facing the analyst is one of optimizing the design matrix.  Without loss of generality, we may assume that $P=I_{N \times N}$, as convexity is preserved by composition with a linear function.  Furthermore, as shown in \cite[Proposition 6.1.]{LewSen05}, a function of the singular values of the form $f(X) = \phi \, \circ \, \sigma(X)$ is convex when $\phi$ is invariant under permutation of its argument and it is convex.  The functions considered above can be rewritten in the following manner, which demonstrates their convexity:
\begin{eqnarray*}
\tr[(X^\top X)^{-1}] &=& \sigma_1^{-2}(X)+\ldots+\sigma_d(X)^{-2}\\
\|(X^\top X)^{-1}\|_{F} &=& \sqrt{ \sigma_1^{-4}(X)+\ldots+\sigma_d(X)^{-4} }\\
\|(X^\top X)^{-1}\|_{\text{op}} &=& \max(\sigma_1^{-2}(X),\ldots,\sigma_d(X)^{-2})\,.
\end{eqnarray*}
We note that \cite[Proposition 6.2.]{LewSen05} also gives a convenient formula for the gradient (or subgradient) of such functions, which is useful in order to solve the associated optimization problems numerically.

\section{Halfspace decision} \label{SEC:halfspace}

We discuss a stylized hypothesis testing problem in order to highlight the applicability of our framework in problems beyond parameter estimation from linear measurements.  Given $c \in \mathbb{R}^d$ and $b \in \mathbb{R}$, the objective for the analyst is to decide whether an unknown $\theta \in \mathbb{R}^d$ is such that
$$\langle \theta , c \rangle > b \, .$$
In our model, the analyst obtains independent information about each coefficient $\theta_i$ of $\theta$ via $d$ independent sources that provide random variables $\hat \theta_1, \ldots, \hat \theta_d$. In this setting, the aggregation step is trivial: $\hat \theta = a(\hat \theta_1, \ldots, \hat \theta_d) = (\hat \theta_1, \ldots, \hat \theta_d)^\top$. The user can expend resource $r_i \geq 0$ on the $i$'th coefficient $\hat \theta_i$ subject to the constraint that $r_1+\ldots+r_d \le R$.  The statistical quality of the random variable $\hat \theta_i$ is governed by a distribution $\bP_{r_i}$ that depends on the resource amount $r_i$ allocated to source $i$.


\subsection{General setup} \label{SUBSEC:generalhalfspace}

Suppose without loss of generality that $\theta$ lies on one side of the hyperplane, with $\langle \theta , c \rangle = b+t$ for some $t>0$. The objective of the problem is to allocate $r_1,\ldots,r_d$ so as to minimize the probability of error
$$\bP_{r_1,\ldots,r_d}\big( \langle \hat \theta , c \rangle \le b \big) =  \bP_{r_1,\ldots,r_d}\big( \langle \theta - \hat \theta , c \rangle \ge t\big) \, .$$
This resource allocation problem is interesting and well-posed when the distribution $\bP_{r_i}$ of $\hat \theta_i$ is more concentrated around $\theta_i$ as $r_i$ increases.  This property can be formalized in a number of ways.  One approach is to require that for all open intervals $I$ that contain $\theta_i$, $\bP_{r_i}(\hat \theta_i \in I)$ must be nondecreasing as a function of $r_i$.  We investigate two specific examples of distributions having this property: in the first case the random variable $\hat \theta_i$ has mean $\theta_i$ and variance decreasing with increasing $r_i$, and in the second case we consider discrete distributions for which $\bP_{r_i}(\hat \theta_i \neq \theta_i)$ is decreasing in $r_i$.

In each of these cases, however, obtaining a closed form expression of $\bP_{r_1,\ldots,r_d}\big( \langle \theta - \hat \theta , c \rangle \ge t\big)$ is a hopeless endeavor in general, and our approach is to minimize an upper bound on this probability:
$$\bP_{r_1,\ldots,r_d}\big( \langle \theta - \hat \theta , c \rangle \ge t\big) \le \Delta_t(r_1,\ldots,r_d) \, .$$
Such upper bounds can be quite sharp in many cases due to the concentration of measure phenomenon, and we seek resource allocation strategies that are based on minimizing $\Delta_t(r_1,\ldots,r_d)$ over a set of possible resources $\cR$.  In the two following subsections, we illustrate this approach by considering cases in which $\theta \in [0,1]^d$ and $\theta \in \{0,1\}^d$.  These examples are motivated by stylized polling resource allocation problems, in which an analyst must decide how best to assign polling resources to states in order to predict the outcome of an election.  The coefficients of $c$ in such scenarios correspond to the weight (e.g., population, electoral votes) of a particular region, and $b$ corresponds to the threshold required for victory.  For simplicity, we assume that there are only two candidates participating in an election and that all voters cast their votes in favor of one of the two candidates.

\subsection{Direct election} \label{SUBSEC:directelec}

In the first example, we consider a direct election setting in a country with $d$ regions.  Here $c_i$ is the voting age population of region $i$, and this region gives $c_i \theta_i$ of its votes to candidate A for $\theta_i \in [0,1]$, i.e., $\theta_i \in [0,1]$ is the (unknown) proportion of candidates who vote for candidate A.  We assume that $\sum_i c_i = 1$ after suitable normalization, and that Candidate A is the winner with $\langle \theta, c \rangle = 1/2 + t$ for $t > 0$.

Of course, as the analyst does not know $\theta_i$ in advance, the goal is to estimate this quantity for each region in order to predict the outcome of the election.  Polling in region $i$ produces an estimate $\hat \theta_i$ of  $\theta_i$.  This estimate has mean $\theta_i$ and variance $\sigma_i^2(r_i)$ as a function of the resource amount $r_i$ allotted to region $i$.  The prediction rule is to declare a victory for candidate A if $\langle \hat \theta , c \rangle > 1/2$.

One can use Bernstein's inequality to obtain a bound on the probability of error of the decision rule:
$$\bP_{r_1,\ldots,r_d}\big( \langle \theta - \hat \theta , c \rangle \ge t\big) \le \exp\Big( - \frac{t^2/2}{ \sum_{i=1}^d c_i^2 \sigma_i^2(r_i) +t\|c\|_{\infty}/3} \Big) \, .$$
The upper bound $\Delta^{\text{dir}}_t(r_1,\ldots,r_d) = \exp\Big( - \frac{t^2/2}{ \sum_{i=1}^d c_i^2 \sigma_i^2(r_i) +t\|c\|_{\infty}/3} \Big)$ on the probability of error is an increasing function of $\sum_{i=1}^d c_i^2 \sigma_i^2(r_i)$, and therefore our resource allocation optimization problem can be expressed as
\begin{equation}
\begin{aligned}
\label{EQ:convexbudgetdirect}
& \text{minimize} &&  \sum_{i=1}^d c_i^2 \sigma^2_i(r_i) &\\
& \text{subject to}&&  r \in \cR \, .
\end{aligned}
\end{equation}
Notice that no prior knowledge of $t$ is needed in order to solve this minimization problem.  In settings in which there are ``diminishing returns'' with the expenditure of additional resources, the variance function $\sigma_i^2(r_i)$ is often well-approximated as being convex and decreasing. In such cases, the problem \eqref{EQ:convexbudgetdirect} is a convex program and can be solved efficiently.  

\subsection{Indirect election} \label{SUBSEC:indirectelec}
An alternative model for elections is the U.S. electoral college system (as well as several other parliamentary systems around the world) in which candidate A is allotted \emph{all} the electoral votes of region $i$ if more than half the voters in region $i$ cast their votes for candidate A.  In this model, for each $i\in [d]$ we have that $\theta_i \in \{0,1\}$.  There is an underlying fraction $\mu_i$ of voters from region $i$ that would vote for candidate $A$, and $\theta_i = \mathbf{1}\{\mu_i >1/2\}$.  The objective of polling in this scenario is to obtain estimates of $\mu_i$, and the nonlinearity associated with going from $\mu_i$ to $\theta_i$ must be taken into account in allocating polling resources to the different regions.


We consider a simplified setup in which the analyst knows a lower bound $\eta_i > 0$ on the margin $|\mu_i - 1/2|$ in advance for each of the regions, i.e., $\eta_i \leq |\mu_i - 1/2|$ for each $i$. Therefore, the analyst has a lower bound on the margin by which candidate A wins or loses a region, but not the precise margin $|\mu_i-1/2|$ (this suffices for our purposes as we minimize an upper bound on the probability of error below). Such information may, for instance, be estimated from past elections; see the numerical experiment in Section~\ref{SUBSEC:pollingexample} for an example.  We make the assumption that polling in each region yields a prediction $\hat \theta_i \in \{0,1\}$ such that $\ell_i(r_i) = \bP_{r_i}(\hat \theta_i \neq \theta_i) \le 1/2$ (i.e., polling yields better results than an unbiased coin flip).  For the sake of illustration, we assume that the probability of error is bounded by
\begin{equation}
\ell_i(r_i) = \tfrac{1}{2} \exp(-r_i \eta_i^2/2). \label{EQ:indirloss}
\end{equation}
(Roughly speaking, this relates to a situation in which polling the region $i$ with resource amount $r_i$ yields an estimate $\hat \mu_i$ of $\mu_i$ with distribution $\cN(\mu_i, 1/r_i)$, which may be viewed as ``polling $r_i$ voters'' in each state.)  These loss functions are known to the analyst since $\eta_i$ (our lower bound on the margin $|\mu_i - 1/2|$) is assumed to be known in advance.

The vector $\hat \theta$ has mean $\tilde \theta$, where $|\tilde \theta_i - \theta_i| = \ell_i(r_i)$, and variance bounded above by $\ell_i(r_i)$. Therefore, we have that
\begin{equation}
|\langle \E[\hat \theta],c \rangle - \langle \theta ,c \rangle| \le \sum_i \ell_i(r_i) c_i =: \beta(r). \label{EQ:bias}
\end{equation}
The quantity $\beta(r)$ can be interpreted as an upper bound on the bias in the polling results, and it is a consequence of the nonlinearity underlying indirect elections.  Suppose there exists $r \in \cR$ such that $\beta(r) < t$, i.e., there is an allocation of resources such that the polling bias is less than the actual advantage of the majority candidate; at the end of this section, we discuss the implications and some potential alternatives if this condition does not hold.  Then the probability of error of a decision rule $\hat \theta$ that predicts the victory of candidate A if $\langle \hat \theta, c \rangle > 1/2$ is bounded as
$$\bP_{r_1,\ldots,r_d}\big( \langle \theta - \hat \theta,c \rangle \ge t\big) \le \bP_{r_1,\ldots,r_d}\big( \langle  \theta- \E[ \hat \theta] ,c\rangle \ge t - \beta(r) \big) \, .$$
Note that determining even the probability on the right-hand-side is a computationally difficult problem in general -- specifically, this question is related to the well-known intractable problem of counting the number of vertices of the hypercube that lie on one side of a given hyperplane \cite{SteVemVig12}.  However, it is possible to obtain further useful upper bounds through Bernstein's inequality, which yields
$$\bP_{r_1,\ldots,r_d}\big( \langle \theta - \hat \theta,c \rangle \ge t\big) \le  \exp\Big( - \frac{(t-\beta(r))^2/2}{ \sum_{i=1}^d c_i^2 \ell_i(r_i) +\|c\|_{\infty} (t-\beta(r))/3} \Big) \, .$$
Consequently, minimizing this upper bound $\Delta^{\text{indir}}_t(r_1,\ldots,r_d) = \exp\Big( - \frac{(t-\beta(r))^2/2}{ \sum_{i=1}^d c_i^2 \ell_i(r_i) +\|c\|_{\infty} (t-\beta(r))/3} \Big)$ on the probability of error can be reformulated as follows, with $\gamma(r):= \sum_{i=1}^d c_i^2 \ell_i(r_i)$:
\begin{equation}
\begin{aligned}
\label{EQ:convexbudget}
& \text{minimize} &&   \frac{2 \, \gamma(r)}{(t-\beta(r))^2} + \frac{4}{3} \frac{\|c\|_{\infty}}{(t-\beta(r))} &\\
& \text{subject to}&&  r \in \cR\, .
\end{aligned}
\end{equation}
If $\cR$ is a convex set, this problem is again a convex program based on \eqref{EQ:indirloss} and \eqref{EQ:bias}.

To reiterate, our reasoning is valid only if there exists an allocation of resources $r \in \cR$ such that $\beta(r) - t < 0$.  If this is not the case, then there is no feasible resource allocation that can reliably predict the victory of candidate A (as the actual advantage of candidate A is $t$); this may, for example, be the case if there are several states with large vote-share $c_i$ and these states also have poor losses $\ell_i(r_i)$ so that a lot of polling resources need to be expended in order to obtain a reliable estimate.  A second issue that arises in practice is that the actual advantage factor $t$ is clearly not known in advance of an election.  Notice that the dependence of the bound $\Delta^{\text{dir}}_t(r_1,\ldots,r_d)$ on $t$ was not a complication in the case of direct elections in Section~\ref{SUBSEC:directelec} (we posed the resource allocation problem \eqref{EQ:convexbudgetdirect} solely in terms of $r$), but in the indirect setting $\Delta^{\text{indir}}_t(r_1,\ldots,r_d)$ is typically dependent on $t$ (even for other choices of $\ell_i(r_i)$ than the one presented here).  One approach to circumvent both these issues is to design a resource allocation for a particular precision $t_d$ that is chosen based on a desired accuracy on $\langle \hat v, c \rangle$.  As long as $t_d > \inf_{r \in \cR} \beta(r)$, it is feasible to minimize $\Delta^{\text{indir}}_{t_d}(r_1,\ldots,r_d)$ by solving the following convex program:
\begin{equation}
\begin{aligned}
\label{EQ:convexbudget2}
& \text{minimize} &&   \frac{2 \, \gamma(r)}{(t_d-\beta(r))^2} + \frac{4}{3} \frac{|c|_{\infty}}{t_d-\beta(r)} &\\
& \text{subject to}&&  r \in \cR \, .
\end{aligned}
\end{equation}

\subsection{A numerical example} \label{SUBSEC:pollingexample}
We consider the problem of allocating polling resources to predict the outcome of the 2016 U.S. presidential election based on data obtained from the results of the 2012 election.  We only count votes cast in favor of the two main candidates in each state and the district of Columbia, and consider these 51 ``regions'' as whole (i.e., we ignore the effect of Nebraska and Maine being able to split their electoral votes).  More broadly, our approach is necessarily simplified and doesn't take into account several other subtleties.  Nonetheless, our numerical results lead to some interesting observations regarding resource allocation problems that arise in inferential settings.

\begin{center}
\begin{tabular}{ |l|r| }
  \hline
  \multicolumn{2}{|c|}{Total resources $R=150,000$} \\
  \hline
  Florida & 38,558.3 \\
  Ohio & 16,448.7 \\
  North Carolina & 14,198.1 \\
  Virginia & 9100.0 \\
  Pennsylvania & 8787.8 \\
  Georgia & 4571.5 \\
  Colorado & 4500.0 \\
  Wisconsin &  3888.7\\
  Minnesota & 3443.6\\
  $\ldots$ & $\ldots$ \\
  Texas & 2098.4 \\
  California &  1495.8\\
  \hline
\end{tabular}
\quad \quad \quad \quad
\begin{tabular}{ |l|r| }
  \hline
  \multicolumn{2}{|c|}{Total resources $R=10,000$} \\
  \hline
  Florida & 9.9 \\
  Ohio & 13.7 \\
  North Carolina & 10.3 \\
  Virginia & 12.9 \\
  Pennsylvania & 387.0 \\
  Georgia & 558.2 \\
  Colorado & 12.6 \\
  Wisconsin &  19.5\\
  Minnesota & 713.5 \\
  $\ldots$ & $\ldots$ \\
  Texas & 1109.7 \\
  California &  1028.4\\
  \hline
\end{tabular}
\end{center}

Our approach to this problem is based on the setup described in Section~\ref{SUBSEC:indirectelec}.  We set $t=63/538$, the actual advantage in the electoral college of the winner of the 2012 election.  We let $\ell_i(r_i) = \tfrac{1}{2} \exp(-r_i\eta_i^2)$, where $\eta_i = |\mu_i-1/2|$ is the actual margin of victory/loss in state $i$ of the winner of the 2012 election, and $c \in \mathbb{R}^{51}$ is the set of normalized electoral votes.  The tables above describe the resource allocations computed using the convex program \eqref{EQ:convexbudget2} for constraint sets $\cR = R \cdot \cU_{51}$, where $R=150,000$ in the first example and $R=10,000$ in the second example.  We observe that when the overall budget is high ($R=150,000$), most of the resources are awarded to so-called ``swing-states'' that have a large number of electoral votes, and for which the vote is almost evenly split between the two main candidates; in other words, $\eta_i$ is close to $0$ (i.e., $\mu_i$ is close to $1/2$) based on the 2012 data.  However, when the analyst only has access to a small overall budget ($R=10,000$), the resources are concentrated on states that have a large number of electoral votes \emph{and} that can be reliably polled with a small amount of resources (for these states $\eta_i$ is far away from $0$, or equivalently $\mu_i$ is further away from $1/2$).  In particular, states that were close calls in 2012 are actually not allocated many resources even if they have a large number of electoral votes.

Hence, these numerical results suggest that there are two regimes: It is only worthwhile to allocate resources to states whose outcome is very hard to determine (the ``too close to call'' states) when there are enough resources available to make a prediction significantly better than a coin flip.  Otherwise, a better strategy is to focus resources on states that have a very high impact on the overall outcome, and to make a very good prediction for those states.



\section{Discussion} \label{SEC:discussion}

We have presented a general framework for the optimal allocation of resources in statistical inference problems involving heterogeneous data sources.  We demonstrate the utility of this framework through several concrete examples.  These illustrations highlight the interplay among different metrics of statistical efficiency, diverse models for the quality of a data source as a function of the resource allocated to it, and various constraints on the manner in which resources can be allocated to different data sources.


Our approach is intentionally general and our examples are idealized in many respects.  However, several refinements that may be of interest in practice could be examined in our framework.  As an example, one could investigate the robustness of the methods described here to imperfect knowledge of the quality of the sources, where the individual loss functions $\ell_i$ are only known within some uncertainty set (similar in spirit to the literature on robust optimization).  In other settings, data sources may not necessarily be independent and the resulting resource allocation questions must take into account any correlations between different sources.  The setup in Section~\ref{SUBSEC:generallinearparam} corresponding to parameter estimation from general linear measurements may be a good preliminary candidate for an extension in this direction, as the resource allocation problems \eqref{EQ:setsoptim1} and \eqref{EQ:setsoptim2} continue to remain tractable even for general convex resource constraint sets $\cP$ rather than just convex subsets of diagonal matrices (recall that $\cP$ specifies a set of resources parameterized by precision matrices).



In a different direction, we only consider regimes in which $n \ge d$ in the setup on parameter estimation from linear measurements in order to avoid ill-posed estimation problems.  The high-dimensional setting where $d>n$ is also of great interest, and generalizing our framework to those situations is an interesting question.  As is common in that literature, additional constraints on the unknown parameter $\theta \in \mathbb{R}^d$ to be estimated could help alleviate the curse of dimensionality, although these must be balanced with the computational consideration that the eventual resource allocation problem must be tractable to solve. \\

\bibliographystyle{amsalpha}
\bibliography{statebib2}

\end{document}